\documentclass[nofootinbib,%
 reprint,
 amsmath,amssymb,
 aps,
]{revtex4-1}

\usepackage{dcolumn}
\usepackage{bm}

\usepackage{mathtools}
\usepackage{amsthm}

\usepackage{graphicx}
\usepackage{caption}
\usepackage{subcaption}
\usepackage{algorithm}
 \usepackage{algpseudocode}
\newtheorem{theorem}{Theorem}
\begin{document}

\preprint{APS/123-QED}

\title{Accessibility and Delay in Random Temporal Networks}

\author{Shahriar Etemadi Tajbakhsh}
\email{shahriar.etemaditajbakhsh@eng.ox.ac.uk}
\author{Justin  P. Coon}%
 \email{justin.coon@eng.ox.ac.uk}

\author{David E. Simmons}%
 \email{david.simmons@eng.ox.ac.uk}

\affiliation{%
 Department of Engineering Science\\
 University of Oxford
}%

%

%

\begin{abstract}
In a wide range of complex networks, the links between the nodes are temporal and may sporadically appear and disappear. This temporality is fundamental to analyze the formation of paths within such networks. Moreover, the presence of the links between the nodes is a random process induced by nature in many real-world networks. In this paper, we study random temporal networks at a microscopic level and formulate the \emph{probability of accessibility} from a node \emph{i} to a node \emph{j} after a certain number of discrete time units $T$. While solving the original problem is computationally intractable, we provide an upper and two lower bounds on this probability for a very general case with arbitrary time-varying probabilities of links' existence. Moreover, for a special case where the links have identical probabilities across the network at each time slot, we obtain the exact probability of accessibility between any two nodes.
Finally, we discuss scenarios where the information regarding the presence and absence of links is initially available in the form of time duration (of presence or absence intervals) continuous probability distributions rather than discrete probabilities over time slots. We provide a method for transforming such distributions to discrete probabilities which enables us to apply the given bounds in this paper to a broader range of problem settings.


\end{abstract}

\pacs{Valid PACS appear here}
\maketitle
\section{Introduction}
The existence of a connection between any pair of nodes in many types of networks is a temporal event and also random in many cases. For instance, human beings meet for some period of time and walk away afterwards.  Because of this temporality, static graphs or even random graphs are incapable of modelling many aspects of random temporal networks. For instance, a \emph{path} between two specific nodes $i$ and $j$ in a static network is a sequence of nodes starting from $i$ and ending at $j$ given that an edge exists between any two successive nodes in the sequence. In a temporal network and over a time window of observation, a sequence of nodes forms a path (or more precisely an \emph{open path} which is defined formally later) if the existence of an edge between two subsequent nodes in the sequence maintains causality. Consider a traveler starting its journey from node $i$ to node $j$. The traveler waits at each node and jumps to the next node from its current node as soon as an edge becomes available between these two nodes. A path exists from $i$ to $j$ if such a traveler reaches $j$ within the observation time window. Therefore, the existence of a temporal path depends on the availability of an edge on or after the current time between its current node and the next node in the sequence, regardless of whether or not an edge had existed before the traveler arrived at this current node.

If there exists at least one temporal path from $i$ to $j$ over a discrete time window of $1, \dots, T$, $j$ is said to be \emph{accessible} from $i$ and such an event is denoted by $i\xrightarrow{T} j$. A directed graph representing the accessibility relation between every pair of nodes in a temporal network is called the \emph{accessibility graph}, where there exists an edge from any node $i$ to any node $j$ if $i\xrightarrow{T} j$. A temporal network over the window $1, \dots, T$ can be modelled as a sequence of $T$ \emph{adjacency} graphs over the time window of observation. Fig. \ref{f_access} shows the adjacency and accessibility graphs for a set of four nodes over three time slots. Obviously the accessibility and adjacency graphs are identical at $t=1$. By $t=2$, node $4$ is accessible from $2$ as there is an edge between $2$ and $1$ at $t=1$ and an edge connects $1$ to $4$ at $t=2$. However, this is not the case for the opposite direction as there is an edge between $4$ and $3$ at $t=1$ but no edge from $3$ to $2$ at $t=2$. Therefore $2$ is not accessible from $4$ by $t=2$. This directionality in the accessibility graph is an immediate consequence of causality in the formation of temporal paths.

An interesting method for obtaining the accessibility graph adjacency matrix (AGAM) in temporal networks is introduced in \cite{lentz2013unfolding}. It should be noted that for a static graph with adjacency matrix $\mathbf{K}_0$, $(\mathbf{K}_0)^T$ gives the number of paths of length at most $T$ between any two nodes, and by changing every non-zero element to $1$ the AGAM is obtained. However, in temporal networks, since the edges between two nodes may appear or disappear at any time, a traveller on the graph might need to \emph{wait} at a specific node for a certain number of time slots until an edge to the next hop becomes available.  Given the adjacency matrices of the adjacency graphs over the window $1,\dots, T$ denoted by $\mathbf{K}_1, \dots, \mathbf{K}_T$, in \cite{lentz2013unfolding}, this waiting at the current node is modelled by adding the identity matrix $\mathbf{1}$ to each adjacency matrix $\mathbf{K}_t$. Therefore calculating $\prod_{t = 1}^T (\mathbf{1}+\mathbf{K}_t)$ and changing all the non-zero elements to $1$ gives the AGAM by time $T$ for such a temporal graph.

The input to the method given in \cite{lentz2013unfolding} is the adjacency matrices. However, in many cases the temporal variations in the edges of the network and their presence and absence is a random process (e.g. wireless ad-hoc networks, human centric networks, etc.).
  In this paper we study the notion of accessibility in random temporal networks. We assume that instead of knowing the adjacency matrices that deterministically identify the presence or absence of an edge between two specific nodes at a certain time slot, we have the probabilities of such events in hand. In other words, a random temporal network is defined as a sequence of random graphs, each one associated with one time slot. Our objective is to obtain the \emph{probability of accessibility}, denoted  by $P(i\xrightarrow{T} j)$, from any node to any other over a window of observation. Since the total number of possible paths between any two vertices grows exponentially and also due to the dependence between paths with common edges, calculation of such probabilities is computationally intractable. In this paper, we provide a non-trivial upper bound and two different lower bounds on these probabilities. Our numerical results show that the accessibility probability obtained by Monte-Carlo simulations of such random temporal networks is very close to the given upper bound.
  Moreover, we examine the upper bound as a predictor for the probability of accessibility over a real-world dataset obtained from a vehicular network (taxis in Rome) \cite{roma-taxi-20140717}. The results show a high correlation between the predicted values and the actual observations.
     Also for the special case where the probability of edges is identical for all the edges across the network in any given time slot (but can vary from one time slot to another), we obtain the exact probability of accessibility.

   It should be noted that in many cases instead of discrete probabilities for the edges over each time slot, the distribution of the duration of the intervals of presence or absence of edges (and mostly in continuous time domain) is available. For instance, the distribution of the inter-contact time between individuals in human centric networks has been studied in the literature   \cite{chaintreau2007impact}. To be able to apply the bounds provided in this paper, these continuous distributions need to be transformed to discrete probabilities of edges \footnote{An edge is assumed to be present between two individual while they are in contact with each other, i.e. when they are within a given proximity of each other.}. In this paper, such transformations are provided. These transformations extend the range of problems that the given bounds are applicable to. Specifically, they provide a general framework for analyzing delay problems in multi-hop networks (e.g. in delay tolerant networks \cite{hong2008routing}) using the bounds obtained in this paper.


  \begin{figure}[t]
  \centering
  \includegraphics[trim = 30mm 52mm 15mm 15mm, clip,scale=.30]{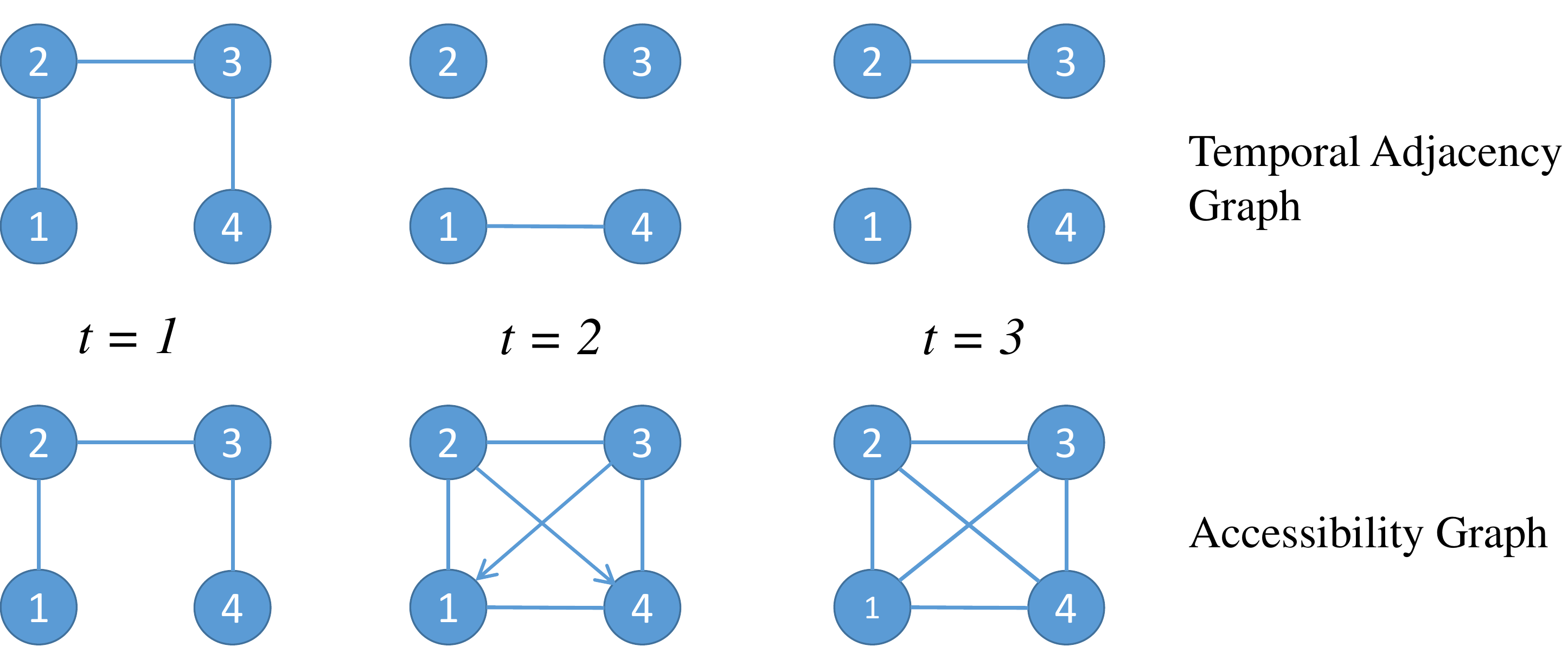}
  \caption{Accessibility Graph} \label{f_access}
\end{figure}

  Current studies in the literature of temporal networks and specifically the notion of accessibility (reachability) can be categorized from different perspectives. Firstly it should be noted that accessibility has been the core of many studies in a wide range of contexts even if the term accessibility (or reachability) has not been explicitly used. Peer-to-Peer networks \cite{qiu2004modeling, stoica2001chord, gkantsidis2005network}, wireless multi-hop networks \cite{de2005high, grossglauser2001mobility}, gossiping over networks \cite{boyd2004gossip, kempe2004spatial, deb2006algebraic}, prevalence of epidemic diseases \cite{eames2002modeling, newman2002spread, keeling2005networks, pastor2001epidemic}, information diffusion in social communication networks \cite{kossinets2008structure, bakshy2012role} or spreading patterns of viruses on smart phones \cite{wang2009understanding} are examples of studies with the theme of accessibility. A very closely related problem to accessibility is delay or trip duration in networks. The duration of a trip is a function of the dynamics of the links. We discuss the relationship between accessibility probability and expected delay (trip duration) in Section \ref{sec:numerical}. An extensive body of research has been devoted to this topic including analysis of the data collected from transportation networks \cite{mokhtarian2004ttb, levinson2005rational}, mathematical modelling of trip durations in human transportation networks at a macroscopic level \cite{kolbl2003energy}, shortest route in time dependant networks \cite{cooke1966shortest} and delay in delay performance of wireless delay tolerant networks \cite{hong2008routing}.

  In a large fraction of studies on path formation and travel duration in temporal networks, methods for measuring various parameters in such networks which are applicable to deterministic (known) temporal networks are proposed \cite{cooke1966shortest, pan2011path, lentz2013unfolding}. However, in this paper we assume that the adjacency evolution of the network graph over the observation time is unknown and only the probabilities of edges' presence between nodes are available. For a special case of such random temporal networks, i.e. the clique network, and from a macroscopic level perspective, the speed of information dissemination is studied in \cite{akrida2014ephemeral}. Another relevant paper to the context of accessibility in random temporal networks is \cite{arenas2010discrete} where the probability of being infected by an epidemic disease is obtained where the individuals are in contact with given probabilities. However, the fundamentally important notion of the dependencies between the paths connecting two nodes caused by common edges between such paths, is ignored.
   In this paper we deal with this dependency and provide an analytical upper and two lower bounds on the accessibility probability. In particular, the theoretical approach used in this paper for obtaining the upper bound is based on a Fortuin-–Kasteleyn-–Ginibre (FKG) correlation inequality \cite{fortuin1971correlation, harris1960lower}, which gives a deeper insight to the problem and provides a basis for analyzing further complex models. Moreover, in our model we consider the very general scenario of time-varying arbitrary probabilities of edges' existence.




 \section{System Model}
 We consider a \emph{random temporal network} with $N$ nodes (vertices) represented by $V=\{1,\dots, N\}$ and a set of discrete time slots $\mathcal{T} = \{0, 1,2, \dots, T\}$.
  There are a total of $M(T) :=N^{T-1}$ vertex sequences of length $T$ between two vertices $i$ and $j$. This is because  at each time we can choose any node in the graph to be the next step (except for the last time slot where node $j$ has been selected).
  The $m$th sequence (possibly with repeated nodes) is represented by
\begin{equation}
A_m^{ij}(T)=v_m^{ij}(0)\dots v_m^{ij}(T),\label{eq:sequence}
\end{equation}
  where $$v_m^{ij}(t)\in \{1,\dots, N\}, \forall t\in \mathcal{T}\backslash \{0, T\},v_m^{ij}(0)=i, v_m^{ij}(T)=j.$$
  Such a sequence of nodes is called a temporal path. An edge between a pair of distinct nodes $u$ and $v$ at time $t$ is denoted by the triple $$(u,v,t),\;\mathrm{where}\; u,v\in \{1,\dots, N\}\;\mathrm{and}\;t\in\mathcal{T}.$$ The triple is defined to be \emph{open} if in the realization of the network the link between these two nodes is physically present. We assume that an edge is open between two nodes $u$ and $v$ with probability $p_{uv}(t)$, independent of other edges.
  A temporal path $A_m^{ij}(T)$ is said to be an open path if any pair of \emph{distinct} successive nodes $v_m^{ij}(t)v_m^{ij}(t+1)$ in the sequence is an open edge.
  In other words we use the terms \emph{open edge} and \emph{open path} adopted from \emph{percolation theory} to identify the realization of an edge or a path. A pair of successive \emph{non-distinct} nodes $v_m^{ij}(t)v_m^{ij}(t+1)$ is an indication of remaining at the same node from time $t$ to $t+1$.

  We use the following compact notation to denote the probability of the event that a given path is open
  $$  P\left(A_m^{ij}(t)\right)\equiv P\left(A_m^{ij}(t) ~is ~open \right).$$ We apply this convention to all the probabilities of the sets corresponding to the temporal paths, including edge triplets $(u,v,t)$.

%
Moreover, we denote a temporal path from $i$ to $j$ with $v_m^{ij}(T-1)=\ell$ by $B_m^{i\ell j}(T)$. The set of all paths inclusive of $\ell$ as their node at time $T-1$ is denoted by $$B^{i\ell j}(T)=\{B_1^{i\ell j},\dots, B_{M(T-1)}^{i\ell j}\}.$$ Our objective is to find the probability that at least one \emph{open temporal path} exists from a given node $i$ to another node $j$ over the time window $\mathcal{T}$.

 \section{Exact Method for Equal Edge Probabilities} \label{sec:exact}
 In this section, we assume that $p_{uv}(t)=p(t)$ (i.e. the probability can change over time but is identical for all the edges in the network at each time $t$). In other words, at each time slot $t$ the network is equivalent to a classic Erd\"{o}s-R\'{e}nyi graph. We start at node $i$ and begin visiting other nodes. Any node $u$ at time $t=1$ is labeled as \emph{visited} if $(i,u,1)$ is open.
 We denote the set of nodes visited for the first time at time slot $t'$ by $\omega(t')$ and the set of all nodes visited from $t = 1$ to $t = t'$ by $W(t')$. Therefore $W(t+1)=W(t)\cup \omega(t+1)$. A node is labeled as visited in time $t$ if there exists an open edge between any node in $W(t-1)$ and this node. Obviously, the total number of visited nodes in $t=1$ is a binomial random variable $B(N-1,p(1))$. If we assume that $|W(t-1)|=\ell$, we have $|\omega(t)|\sim B(N-1-\ell, 1-(1-p(t))^{\ell})$. Therefore, we can conclude that
\begin{align}
&P(|W(t)|=k)=\sum_{\ell=0}^k P(|W(t-1)|=\ell)\dbinom{N-1-\ell}{k-\ell}\nonumber\\
&\times (1-(1-p(t))^{\ell})^{k-\ell}(1-p(t))^{N-1-k}\label{eqq1}
\end{align}
 \begin{figure}[t]
  \centering
  \includegraphics[trim = 25mm 25mm 15mm 5mm, clip,scale=.32]{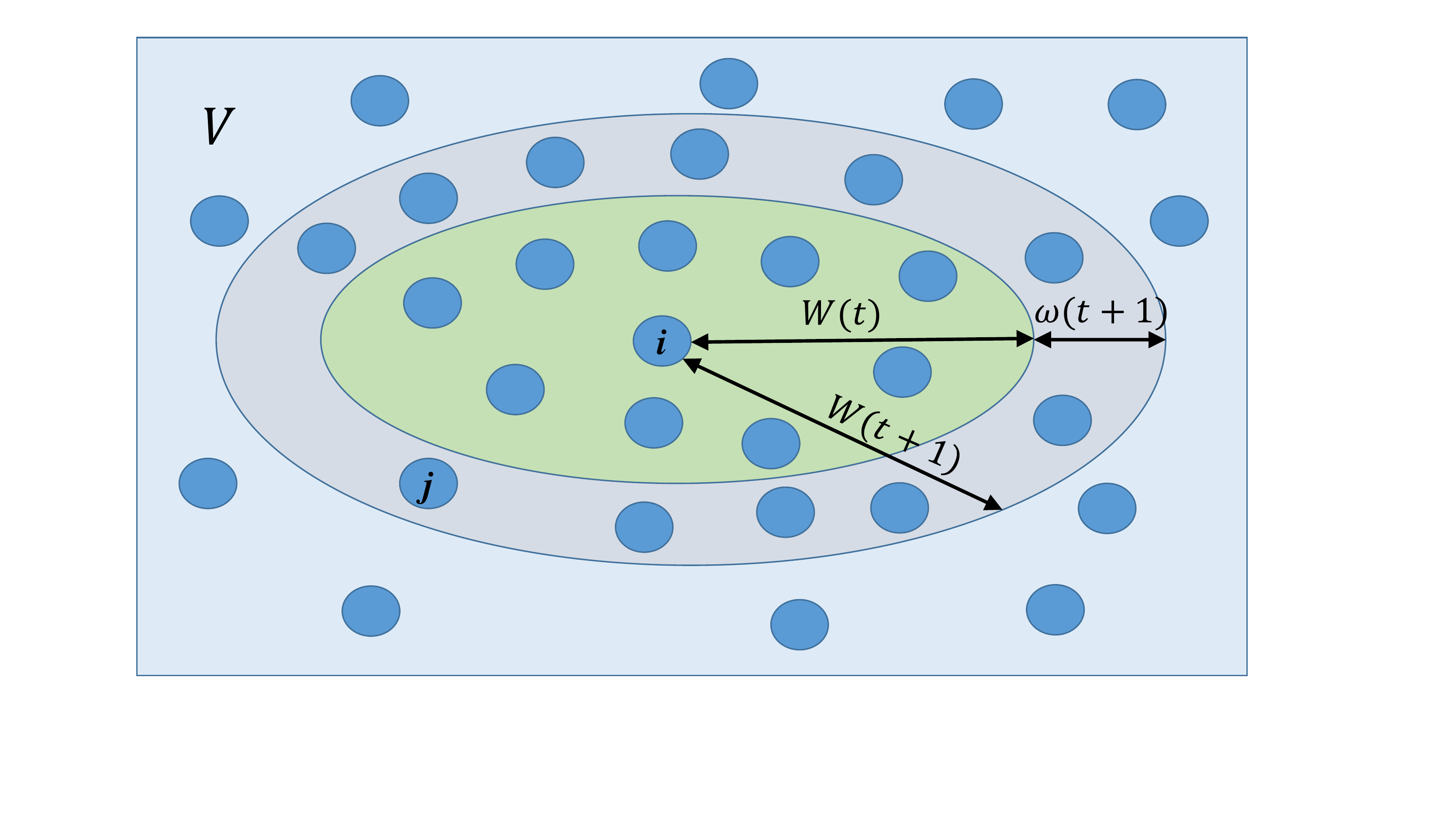}
  \caption{Node $j$ has not been visited by time $t$. In time slot $t+1$ node $j$ falls into the set of nodes labeled as visited.} \label{fig:exact}
\end{figure}

The probability $P(i\xrightarrow{T} j)$ is equivalent to the probability of $j$ being labeled as visited by time $T$ (see Fig. \ref{fig:exact}). Therefore,
\begin{align}
P(i\xrightarrow{T} j)&=P(j\in W(T))\nonumber\\
&=\sum_{\ell=1}^N P(j\in W(T)\Big| |W(T)|=\ell)P(|W(T)|=\ell)\nonumber\\
&=\sum_{\ell=1}^N \frac{\ell}{N-1}P(|W(T)|=\ell),\label{eqq2}
\end{align}
and one can obtain $P(i\xrightarrow{T} j)$ recursively.

\section{Upper Bound}\label{sec:upper}
Generalizing the exact method in Section \ref{sec:exact} to the case of arbitrary time varying probabilities is not straightforward; and even if possible, it would be computationally intractable. Therefore, we propose a different method in this section, which provides an upper bound on $P(i\xrightarrow{T} j)$ given that the probabilities $p_{ij}(t)$ can take any value (between $0$ and $1$) at each time slot. The event of at least one open path existing from node $i$ to $j$ is the complement of the event that no open paths exist from $i$ to $j$. Thus, obtaining the probabilities of every path from $i$ to $j$ should give the desired accessibility probability. However, it should be noted that, firstly, even finding the probability of one path is not straightforward because the number of trials at each node to jump to the next node (waiting time at each node) is a random variable by itself and in general, of a different probability distribution; secondly, the number of paths from $i$ to $j$ grows exponentially with $t$; thirdly and most importantly, different paths might be \emph{positively correlated} if they have any edge in common in the same time slot. In the following, the derivation of the upper bound is discussed.

We associate a dependant variable $\alpha_{ij}(t)$ to any pair of nodes $(i,j)$. This variable is formed recursively and given as follows
\begin{equation}
\begin{split}
&\alpha_{ij}(t)=1-\prod_{\ell=1}^{N}\big(1-\alpha_{i\ell}(t-1)p_{\ell j}(t)\big)\\
&\alpha_{ij}(1)=p_{ij}(1).\label{eq:alphadef}
\end{split}
\end{equation}
In the following theorem we show that $\alpha_{ij}(t)$ is an upper bound for the probability of an open temporal path existing from any node $i$ to any node $j$.
\begin{theorem}
  $P\left(i\xrightarrow{T} j\right)\leqslant \alpha_{ij}(T)$, for all $(i,j)\in V\times V$ and any positive integer $T\geqslant 1$.
\end{theorem}
\begin{proof}
We use induction by showing that $P(i\xrightarrow{T+1} j)\leq \alpha_{ij}(T+1)$ given that $P(i\xrightarrow{T} j)\leq \alpha_{ij}(T)$. Obviously the theorem holds for $T=1$ because $\alpha_{ij}(1)=p_{ij}(1)$ by definition~\eqref{eq:alphadef}.
At time $T$, we have
$$P(i\xrightarrow{T} \ell)
 = P\big(\bigcup_{m=1}^{M(T)}A_m^{i\ell}(T)\big)\leqslant \alpha_{i\ell}(T).$$
where the inequality follows from the induction assumption. This implies that
\begin{equation}\label{indassum}
\begin{split}
P\big(\bigcup_{m=1}^{M(T)}A_m^{i\ell}(T)\big)p_{\ell j}(T+1)\quad\leqslant \alpha_{i\ell}(T)p_{\ell j}(T+1).
\end{split}
\end{equation}
Since the existence of open edges between nodes are independent random variables, we have:
\begin{equation} \label{eqqqq}
\begin{split}
P\Big(\bigcup_{m=1}^{M(T)}&A_m^{i\ell}(T)\Big)p_{\ell j}(T+1)\\
 &=P\Big(\big(\bigcup_{m=1}^{M(T)}A_m^{i\ell}(T)\big)\cap (\ell, j, T+1)\Big)\\
&=P\Big(\bigcup_{m=1}^{M(T)}\big(A_m^{i\ell}(T)\cap  (\ell, j, T+1)\big)\Big) \\
&= P\Big(\bigcup_{m=1}^{M(T)}B_m^{i\ell j}(T+1)\Big).
\end{split}
\end{equation}
Combining \eqref{indassum} and \eqref{eqqqq}, we have
\begin{equation}\label{eq:res}
\begin{split}
&P\Big(\bigcup_{m=1}^{M(T)}B_m^{i\ell j}(T+1)\Big)\leqslant \alpha_{i\ell}(T)p_{\ell j}(T+1) \\
&\Rightarrow P\Big(\bigcap_{m=1}^{M(T)}\overline{B}_m^{i\ell j}(T+1)\Big)\geqslant 1-\alpha_{i\ell}(T) p_{\ell j}(T+1),
\end{split}
\end{equation}
where $\overline{B}_m^{i\ell j}(T+1)$ is the complement of the event $B_m^{i\ell j}(T+1)$.
Each set of paths $\bigcup_{m=1}^{M(T)}B_m^{i\ell j}(T+1)$ for $\ell=1,\dots, N$ is a family of monotonically decreasing events\footnote{A Family $\mathcal{A}$ of subsets of $K = \{1,2,\dots,k\}$ is monotone decreasing if $A\in\mathcal{A}$ and $A'\subseteq A \Rightarrow A'\in \mathcal{A}$. The collection of any open path (viewed as an edge set) and all its subpaths (also viewed as edge sets), form a monotonically decreasing family. This is because if a path is open (and consequently in the family of open paths), any subpath would also be open and hence an element of the family. Here, to avoid unnecessary complication in the notations, we have used $B_m^{i\ell j}(T+1)$ to represent such a family of events.}. Therefore, using the Harris-FKG inequality (Theorem 6.3.2 in \cite{alon2004probabilistic}) we can establish the proof
\begin{equation}
\begin{split}
P\Big(\bigcup_{m=1}^{M(T+1)}&A_m^{ij}(T+1)\Big)\\
&=P\Big(\bigcup_{\ell=1}^N \big(\bigcup_{m=1}^{M(T)}B_m^{i\ell j}(T+1)\big)\Big) \\
&= 1-P\Big(\bigcap_{\ell = 1}^N \big(\bigcap_{m=1}^{M(T)}\overline{B}_m^{i\ell j}(T+1) \big)\Big)\\
&\leqslant 1-\prod_{\ell=1}^N P\big(\bigcap_{m=1}^{M(T)}\overline{B}_m^{i\ell j}(T+1)\big)\\
&\leqslant 1-\prod_{\ell=1}^N \big(1-\alpha_{i\ell}(T)p_{\ell j}(T+1)\big)\\
&= \alpha_{ij}(T+1),
\end{split}
\end{equation}
where the first inequality follows from the Harris-FKG inequality and the second inequality follows immediately from \eqref{eq:res}.
\end{proof}

\section{Lower Bounds}\label{sec:lower}
In this section we provide two alternative lower bounds on $P(i\xrightarrow{T} j)$. The performance of each bound depends on the distribution of the probabilities of edges across the network and over the observation window.

\subsection{Lower Bound I}\label{sec:lower1}
The first lower bound on $P(i\xrightarrow{T} j)$ relies on finding a clique graph inside the temporal network such that the probability of all the edges in this clique is above a certain threshold over the entire window of observation. More formally, we find a subset of nodes $\hat{V}\subseteq V$ for a fixed value $p_{min}$ such that $i,j \in \hat{V}~\mathrm{and}~ p_{\ell m}(t)\geq p_{min}, \forall \ell, m \in \hat{V}, t\in \{1,\dots, T\}.$

For any such $\hat{V}$ and $p_{min}$ we generate a new temporal network $G_{\hat{V}}$ for which we set the probability of all edges to be $p'_{\ell m}(t)=p_{min}, \forall \ell ,m \in \hat{V}, t\in \{1,\dots, T\}$.
Since the probabilities of edges are assumed to be identically $p_{min}$, we can apply the method in Section \ref{sec:exact} to the resulting temporal network with the vertex set $\hat{V}$ and find the probability $\beta_{ij}(T) := \hat{P}(i\xrightarrow{T} j)$, where we use the notation $\hat{P}$ (as opposed to $P$) to distinguish the probability of accessibility in the derived network (defined by vertex set $\hat{V}$ and identical probability $p_{min}$) from the probability of accessibility in the original network.

From the construction of $G_{\hat{V}}$, it is easy to conclude  $\beta_{ij}(T)  \leqslant P(i\xrightarrow{T} j)$ in the original network. Therefore any such subset $V'$ and $p_{min}$ provides a lower bound for the probability of accessibility between two nodes $i$ and $j$. Such a clique is not unique as it depends on the choice of $p_{min}$. It should be noted that the size of a clique by itself does not identify the bound. For instance, a smaller clique with a higher value of $p_{min}$ might result in a higher probability of accessibility and hence a better bound. Therefore, different values of $p_{min}$ should be examined, and the highest accessibility probability should be selected as the lower bound. Based on this heuristic, for a fixed $p_{min}$, we obtain the corresponding lower bound as follows (see also Fig. \ref{fig:clique}):

\begin{figure*}[t]
  \centering
  \includegraphics[trim = 15mm 75mm 15mm 35mm, clip,scale=.54]{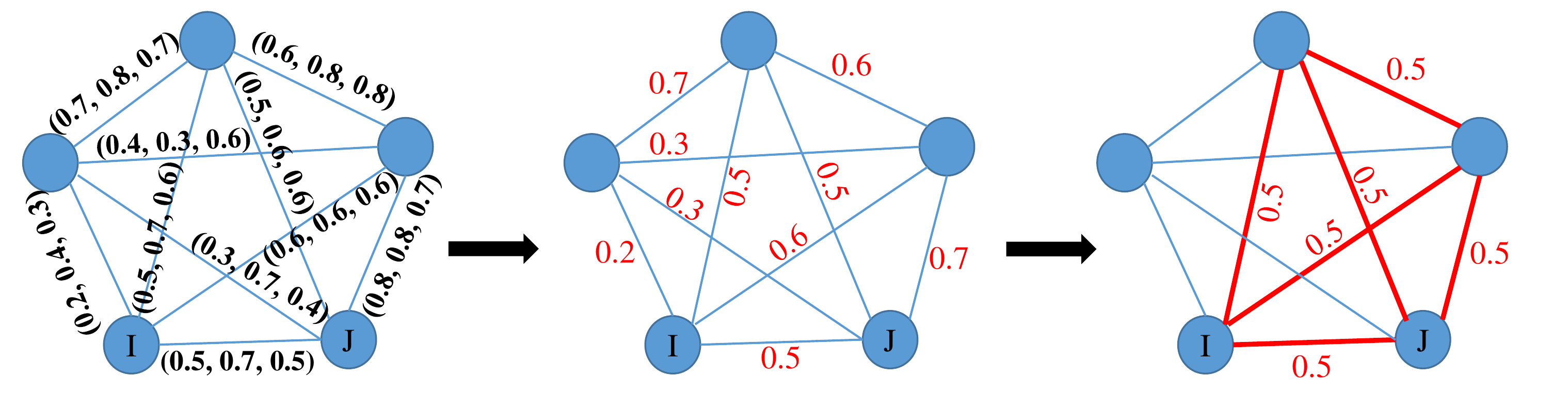}
  \caption{Lower Bound I: Searching for a clique based on a given threshold $p_{min} = 0.5$. The triplet on each edge represents the probabilities of links' be open in time slots $t = 1$, $t = 2$ and $t = 3$.} \label{fig:clique}
\end{figure*}

\begin{itemize}
\item \emph{Step 1}: Find the sets $E_c = \{(\ell,m): p_{\ell m} (t)\geqslant p_{min},  1\leqslant t \leqslant T \}$ and $V_c = \{\ell: \exists m \mathrm{~such~that~} p_{\ell m}(t) \geqslant p_{min}, 1\leqslant t \leqslant T\}$. Form the corresponding equivalent static graph $G_c = (V_c, E_c)$.
\item \emph{Step 2}: Using Bron-Kerbosch \cite{bron1973algorithm} algorithm find the set of all maximal cliques of $G_c$.
\item \emph{Step 3}: Select the largest clique $\hat{V}$ from the subset of cliques that includes $(i,j)$. Generate the temporal network $G_{\hat{V}}$ according to the selected clique, such that $p_{\ell m}(t) = p_{min},\forall \ell,m \in \hat{V}$.

\item \emph{Step 4}: Apply the method in Section \ref{sec:exact} to $G_{\hat{V}}$ and find $\hat{P}(i\xrightarrow{T} j)$.
\end{itemize}
By examining the above method and finding the maximum value of $\beta_{ij}(T)$ (denoted by $\beta_{ij}^*(T)$) for different values of $p_{min}$ \big(possibly for all values of $p_{min} = p_{\ell m}(t): p_{\ell m}(t)\leq \min_t\{p_{ij}(t)\}$\big) one can obtain $$\beta_{ij}^*(T)\leqslant P(i\xrightarrow{T} j).$$

\subsection{Lower Bound II}
The second lower bound is established based on finding a set of \emph{edge-disjoint paths}, because the events of such paths being open are statistically independent. This independence results in finding the probability of at least one path (from the selected edge-disjoint subset of all the paths) being open, without being concerned about the correlation between the paths (as there is no common edge between any two paths within the selected set). It is worthwhile to mention that in general the definition of an edge disjoint path in a temporal network is different from its equivalent in static graphs. In temporal networks an edge $(i,j,t)$ is not only identified by its two end nodes $i$ and $j$ but also with a label of time $t$. Two edges $(i,j,t)$ and $(i',j',t')$ are disjoint if $i\neq i'$ or $j\neq j'$ or $t\neq t'$. Therefore two edge disjoint paths might share the link between two nodes with two distinct time labels.

If we denote the waiting period at each node on a path $$R_k= i, v_1, \dots, v_{L_R-1}, j$$ of length $L_R$ by $t_1,\dots, t_{L_R}$, the existence of an open path between $i$ and $j$ implies that $t_1+\dots +t_{L_R} \leqslant T$. The waiting time at each node is the number of time slots from the current time slot to the time slot in which an edge to the next node in the path becomes open.
Clearly finding the probability of all sequences of numbers with summation less than $T$ given the probabilities of all the edges on this path (which are possibly time-variant as well), is computationally costly. However, if we set $$p_{min}= \min_{R_k} \min_t \{p_{iv_1}(t), p_{v_1v_2}(t), \dots, p_{v_{L_R-1}j}(t)\},$$ then we can assume all the edges have at least a success probability $p_{min}$. Therefore, $P(R_k(T))$, the probability of $R_k$ being open before time $T$, can be lower bounded by the following inequality
  \begin{equation}
  P(R_k(T))\geqslant 1-\sum_{m=1}^{L_R-1} \dbinom {T}{m} p_{min}^m (1-p_{min})^{T-m}.
  \end{equation}
  The inequality holds because the probability of $R_k$ (of length $L_R$)  being open is translated to the probability of observing at least $L_R$ successful outcomes of $T$  Bernoulli trials with parameter $p_{min}$.

  In what follows, it will be helpful to define the \emph{quality of path} to be
   \begin{equation}
   f(R_k(T)) = 1-\sum_{m=1}^{L_R-1} \dbinom {T}{m} p_{min}^m (1-p_{min})^{T-m}.\label{eq:pathquality}
    \end{equation}
    Our objective is to find edge disjoint paths with high qualities to obtain a tighter lower bound for $P(i\xrightarrow{T} j)$ (However, any set of disjoint paths would result in a lower bound). Fig. \ref{fig:disjoint} compares two paths for their quality $f(R_k(T))$. As it can be observed, a longer path (of length $3$) can have a higher quality than a shorter path (of length~$2$).
\begin{figure}[t]
  \centering
  \includegraphics[trim = 15mm 65mm 10mm 30mm, clip,scale=.30]{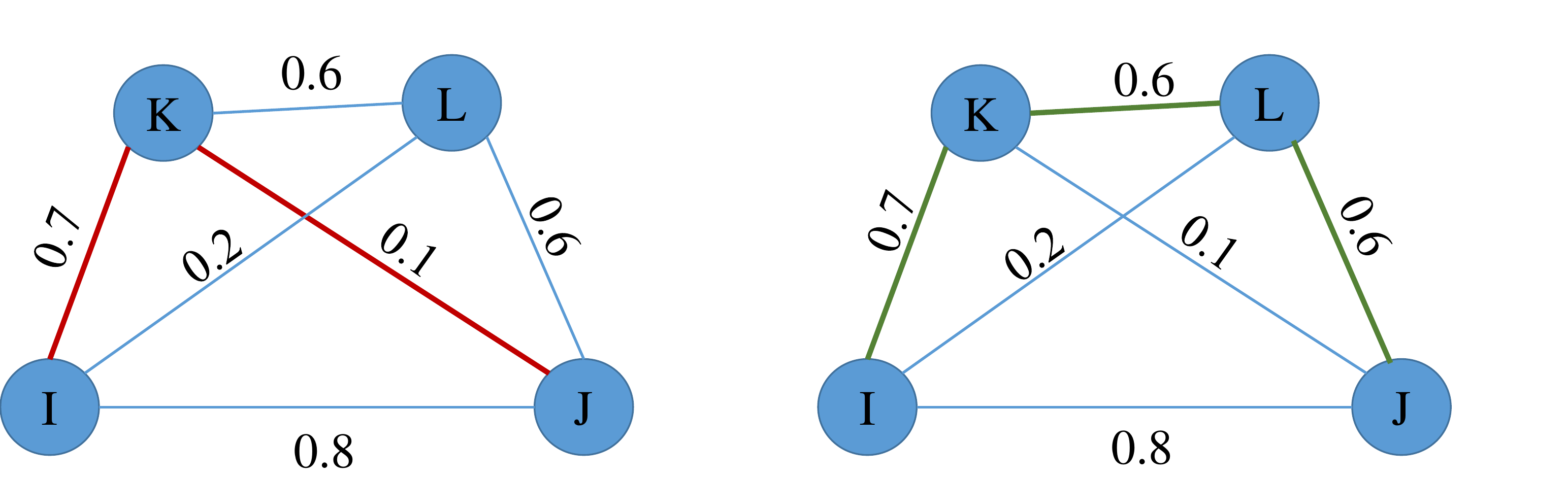}
  \caption{Comparing two paths $R_a = I,K,J$ and $R_b = I,K,L,J$ based on their quality~\eqref{eq:pathquality}, $f\big(R_a(10)\big) = 0.26$ (on the left) and $f\big(R_b(10)\big) = 0.98$ (on the right)} \label{fig:disjoint}
\end{figure}

To obtain the lower bound, we form a set of edge disjoint paths. This is done by first forming a random graph $G_{min}$ where we set the probability $p_{\ell m} = \min_t\{p_{\ell m}(t)\}$ for any pair of nodes in the network. We sort the outgoing edges from $i$ such that $p_{il_1}\geqslant p_{il_2} \geqslant \dots$. We start from $(i, l_1)$ to form the first path i.e. $R_1(T)$. Initially this path would be $i,l_1, j$. At each step the quality of path $f(R_1(T))$ is measured, the outgoing edges from the last node (the node before $j$) are sorted and the edge with maximum probability (between the current last node and a node $m$) is selected. If adding node $m$ as the last node before arriving at $j$ increases the quality of the path, the discovered path is updated by adding $m$ to the path just before $j$. For instance, in Fig. \ref{fig:disjoint}, for $T = 10$, we initially examine the path $I, K, J$ and then we update the path to $I, K, L, J$ as the latter has a higher quality. We stop adding nodes to the path as soon as the quality of the path begins to decrease. Once the path $R_1(T)$ is formed, we remove all the edges of this path from $G_{min}$ and we start generating a new path by repeating the same procedure on $G_{min}$ starting from the next outgoing edge from $i$ in the sorted list of edges. We continue this algorithm until we cannot generate any new path from $i$ to $j$. This algorithm of path selection is represented in Algorithm \ref{myalgorithm}.
If we denote the set of generated paths by $\mathcal{R}$, the lower bound II is given as follows
\begin{equation}
P(i\xrightarrow{T} j)\geqslant 1-\prod_{k=1}^{|\mathcal{R}|} \big(1-f(R_k(T))\big) = \gamma_{ij}(T).
\end{equation}
\begin{algorithm}[H]
\caption{Path Selection}\label{myalgorithm}
\begin{algorithmic}
\While {$\exists ~path  \in G_{min}$}
   \State $current~node\gets i$
   \State $next~node \gets \arg\max_{\ell}p_{i\ell}, \ell \neq i$
   \State $k \gets 1$
   \State $R'_k \gets i,next~node$
   \State $R_k \gets R'_k, j$
   \State $g \gets 0$
  \While {$f\big(R_k(T)\big)\geqslant g$}
     \State $g \gets f\big(R_k(T)\big)$
     \State $current~node \gets next~node$
     \State $next~node \gets \arg\max_{\ell}p_{current~node,\ell},~\ell \neq ~current~node$
     \State $R'_k \gets R'_k,next~node$
     \State $R_k \gets R'_k, j$
  \EndWhile
  \State $k \gets k+1$
\EndWhile
\end{algorithmic}
\end{algorithm}



\section{Quantization of Continuous Random ON-OFF Links}\label{sec:ON-OFF}
In many scenarios, a temporal network evolves in a continuous time fashion. Said in a different way, we may have the probability distribution of the ON or OFF intervals in the continuous-time domain. The ON period refers to the time interval that there exists an open edge between two specific nodes.  By splitting the entire observation window into time slots of a given length, and deriving the probability of the state of an edge being observed at the end of each time slot, we can quantize such continuous distributions.
Such a quantization associates a discrete probability to each time slot.
 This enables us to use the bounds given above to estimate the probability of accessibility between any two nodes in the network.

The link between any two nodes could have started from being ON or OFF at $t=0$ and could have been switched ON and OFF any $m$ number of times between $m=0$ to $m=\infty$ during the interval $t = 0$ to $t = T_0$ (see Fig. \ref{fig:ON-OFF}). Therefore, an infinite number of possible events should be considered when deriving the probability of observing the edge in the ON position.

If we are given the probability distributions of the ON and OFF periods (for a specific edge) denoted by $f_{ON}(\tau)$ and $f_{OFF}(\tau)$ and also the probability of starting from the ON position at $t = 0$ (represented by $p_0$), we can obtain the probability of being in the ON position (denoted by $SW=1$, and $SW=0$ for OFF)  at time $T_0$ which is derived as follows:

\begin{figure}[t]
  \centering
  \includegraphics[trim = 0mm 0mm 135mm 0mm, clip,scale=.40]{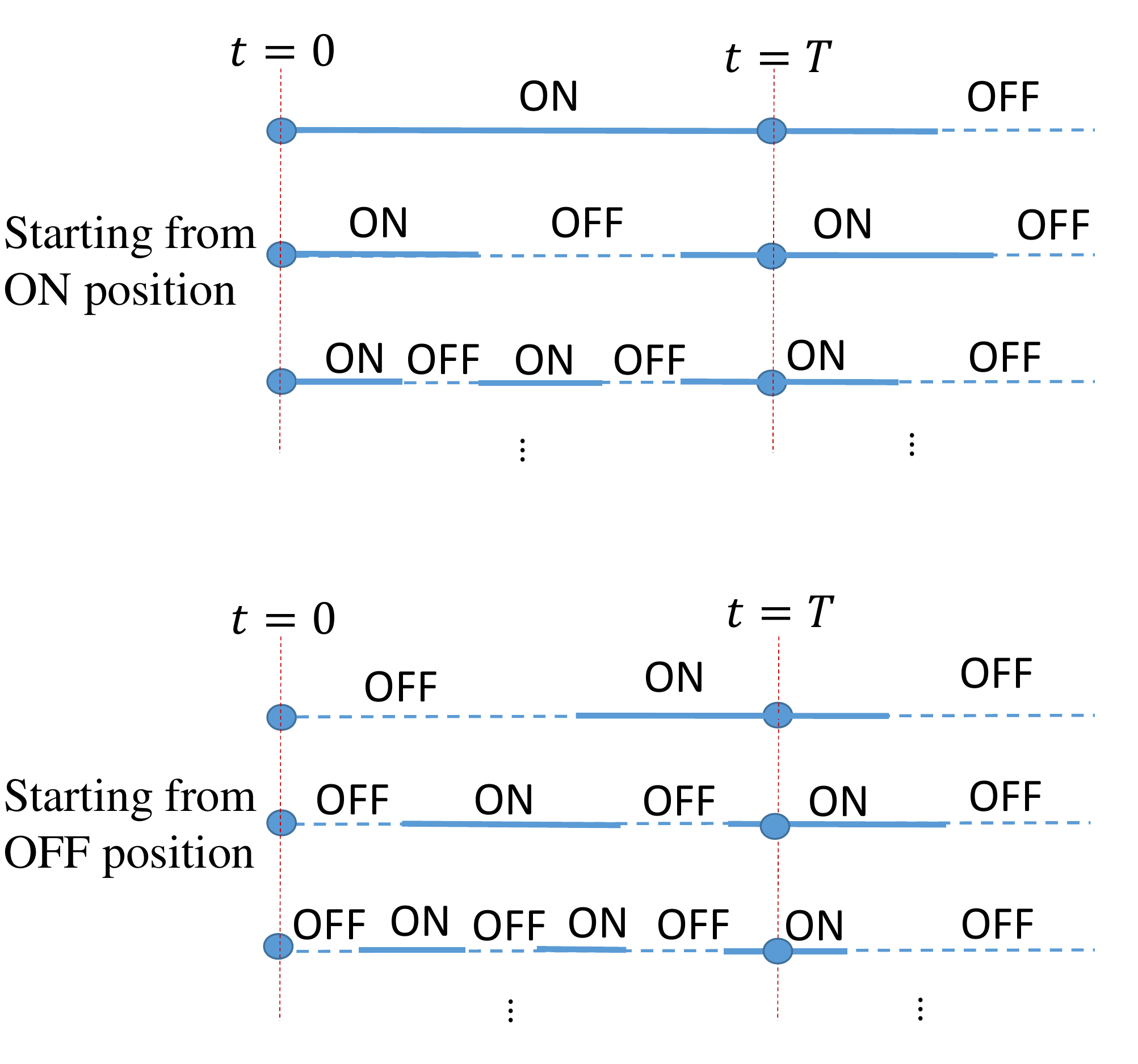}
  \caption{Different possibilities (events) based on starting from ON or OFF position for a given edge} \label{fig:ON-OFF}
\end{figure}

\begin{equation}\label{eq:ON-OFF}
\begin{split}
&P(SW=1)= p_0\sum_{m=0}^{\infty}\int_{0}^{T_0} f_{S_m^{ON}}(s)\big(1-F_{ON}(T_0-s)\big) ds \\
&+(1-p_0)\sum_{m=0}^{\infty}\int_{0}^{T_0} f_{S_m^{OFF}}(s)\big(1-F_{ON}(T_0-s)\big) ds
\end{split}
\end{equation}
where,

\[
f_{S_m^{ON}}(s)
     \begin{cases}
       1 & m=0\\
      \underbrace{f_{ON}\ast \dots \ast f_{ON}}_{m}\underbrace{\ast f_{OFF} \dots \ast f_{OFF}}_{m} & m>0 \\
 \end{cases}
\]

$f_{S_m^{OFF}}(s) = \underbrace{f_{ON}\ast \dots \ast f_{ON}}_{m}\underbrace{\ast f_{OFF} \dots \ast f_{OFF}}_{m+1}$

\noindent and $F_{ON}$ is the cumulative distribution function (CDF) of the ON time distribution. Also $S_m^{ON}$ represents the random variable describing the sum of $m$ periods of ON-OFF (an ON period followed by an OFF period), assuming that at $t = 0$ the edge has been $ON$. Similarly, $S_m^{OFF}$ represents a similar sum with the assumption that at $t = 0$, the edge has been OFF. Starting from ON (with probability $p_0$) or OFF (with probability $1-p_0$) at $t = 0$ are two mutually exclusive events. For each of these two possibilities the summation of the probabilities of infinitely many exclusive events are calculated as in Equation (\ref{eq:ON-OFF}). Each of these many events corresponds to a given $m$ number of switchings (ON-OFF periods). To derive the probability of such an event (i.e. being in an ON position at $t = T_0$ starting from an ON position at $t = 0$ with $m$ switchings)  one needs to obtain the probability of the event that $S_m^{ON}\leqslant T_0$ (the total duration of the $m$ ON-OFF periods is less than $T_0$) and after these $m$ periods an ON period lasts at least until $t = T_0$ and possibly onwards.
Similar arguments apply to the case of starting from the OFF position at $t = 0$, which is reflected in the integrals within the second summation in Equation (\eqref{eq:ON-OFF}).
The integral terms inside the summations give the probability of such events. Since the last ON period can take any value larger than $T_0-S_m^{ON}$, its probability is $1-F_{ON}(T_0-s)$.

It should be noted that this calculation provides the probability of an edge being open if it is observed at a specific instance of time $T_0$. However, another useful probability would be the probability of observing at least one ON period between $t = 0$ and $t = T_0$ to be reported as the probability of the edge being open over the corresponding time slot. Moreover, in practice one can choose the quantization step of time ($T_0$) very small such that with high probability at most one switching occurs within each time slot to ease the calculations. We skip the details of these last two possibilities. However, we apply these simplifications to the numerical experiment in Section \ref{real}.

\section{Numerical Experiments}\label{sec:numerical}
The bounds and mapping from continuous interval distributions to discrete probabilities of edges introduced in this paper have been examined against synthetically generated networks (Monte Carlo simulations) as well as a real-world vehicular network dataset \cite{roma-taxi-20140717} collected from the GPS devices of the taxis in Rome, Italy. In the following, these experiments and their results are reported.
\subsection{Synthetically Generated Networks}\label{sec:monte_carlo}
In our first experiment, we considered a network with $N = 20$ nodes and the probability of any edge $(\ell, m)$ was selected randomly from the range $[0.05~~ 0.1]$. However, we fixed this probability for the observation window $t = 1,\dots, 15$. Fig. \ref{fig:accessibility} shows the evolution of the the probability of accessibility $P(i \xrightarrow{t} j)$ for a randomly selected pair of nodes $i$ and $j$ over the mentioned window. As it can be observed from the figure, the upper bound is very close to the expected probability of accessibility obtained from Monte-Carlo simulations\footnote{For convenience, in the presentation we have shown the probability distributions over a continuous domain by interpolation in Fig \ref{fig:accessibility} and \ref{fig:delay}.}. It should be noted that the Lower Bound I is computationally slow because of the clique search algorithm; however the other bounds, and particularly the upper bound, are easily applicable to large-scale networks as well.

\begin{figure}[t]
  \centering
  \includegraphics[trim = 0mm 96mm 0mm 50mm, clip,scale=.40]{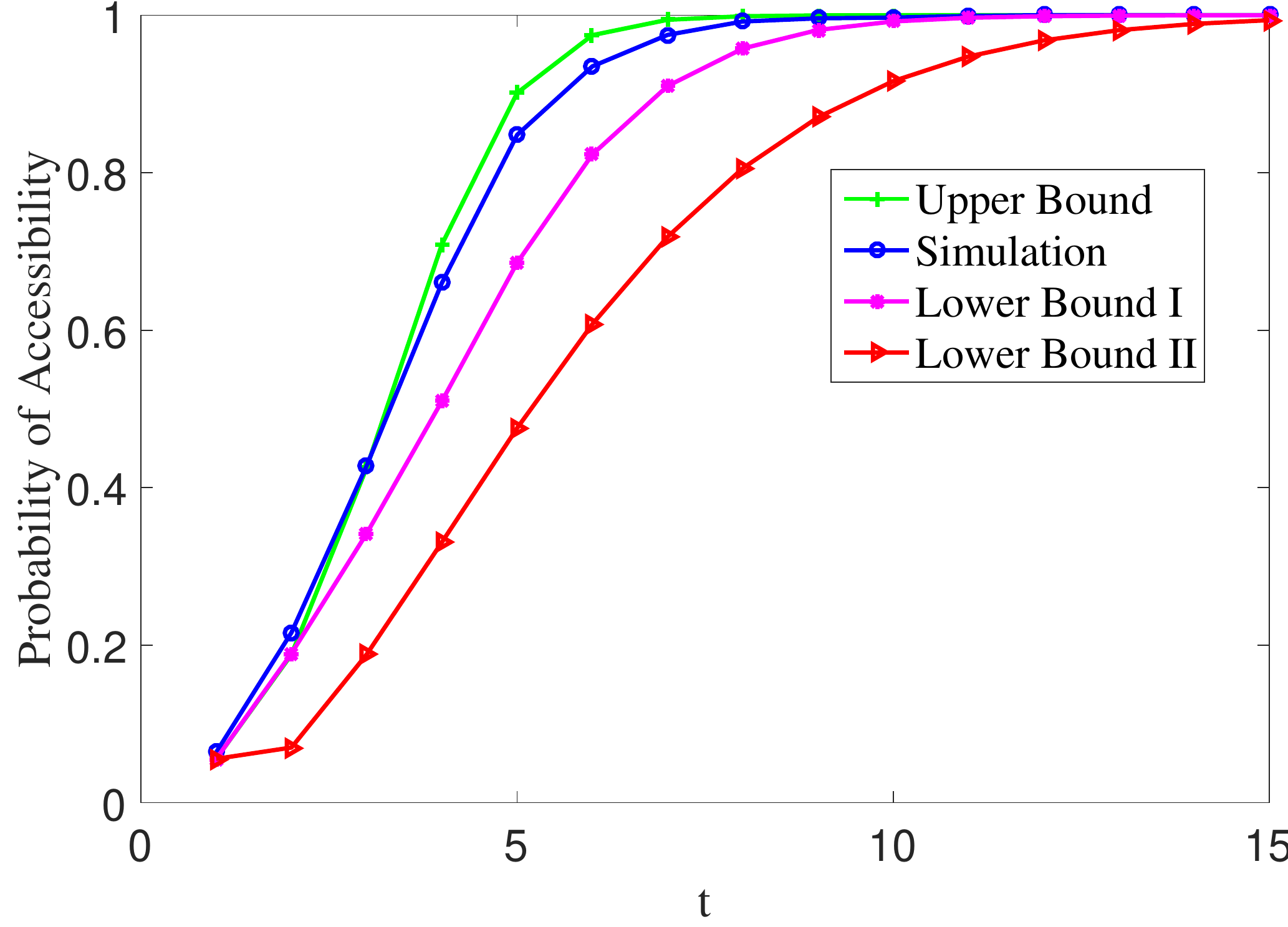}
  \caption{Probability of accessibility between two specific nodes} \label{fig:accessibility}
\end{figure}

\begin{figure}[t]
  \centering
  \includegraphics[trim = 5mm 95mm 20mm 70mm, clip,scale=.48]{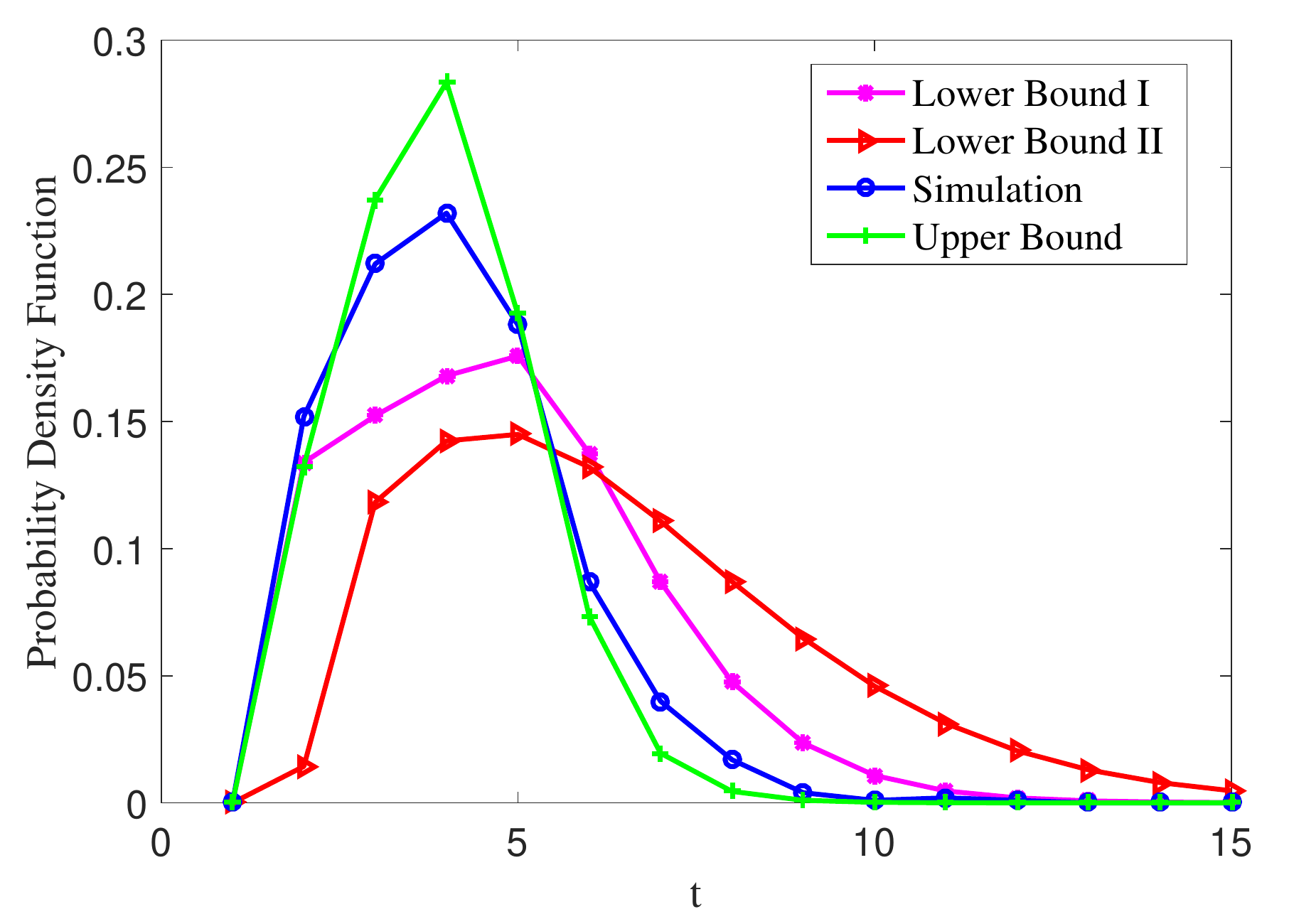}
  \caption{Accessibility delay distribution} \label{fig:delay}
\end{figure}

\begin{figure}[t]
  \centering
  \includegraphics[trim = 15mm 75mm 10mm 80mm, clip,scale=.47]{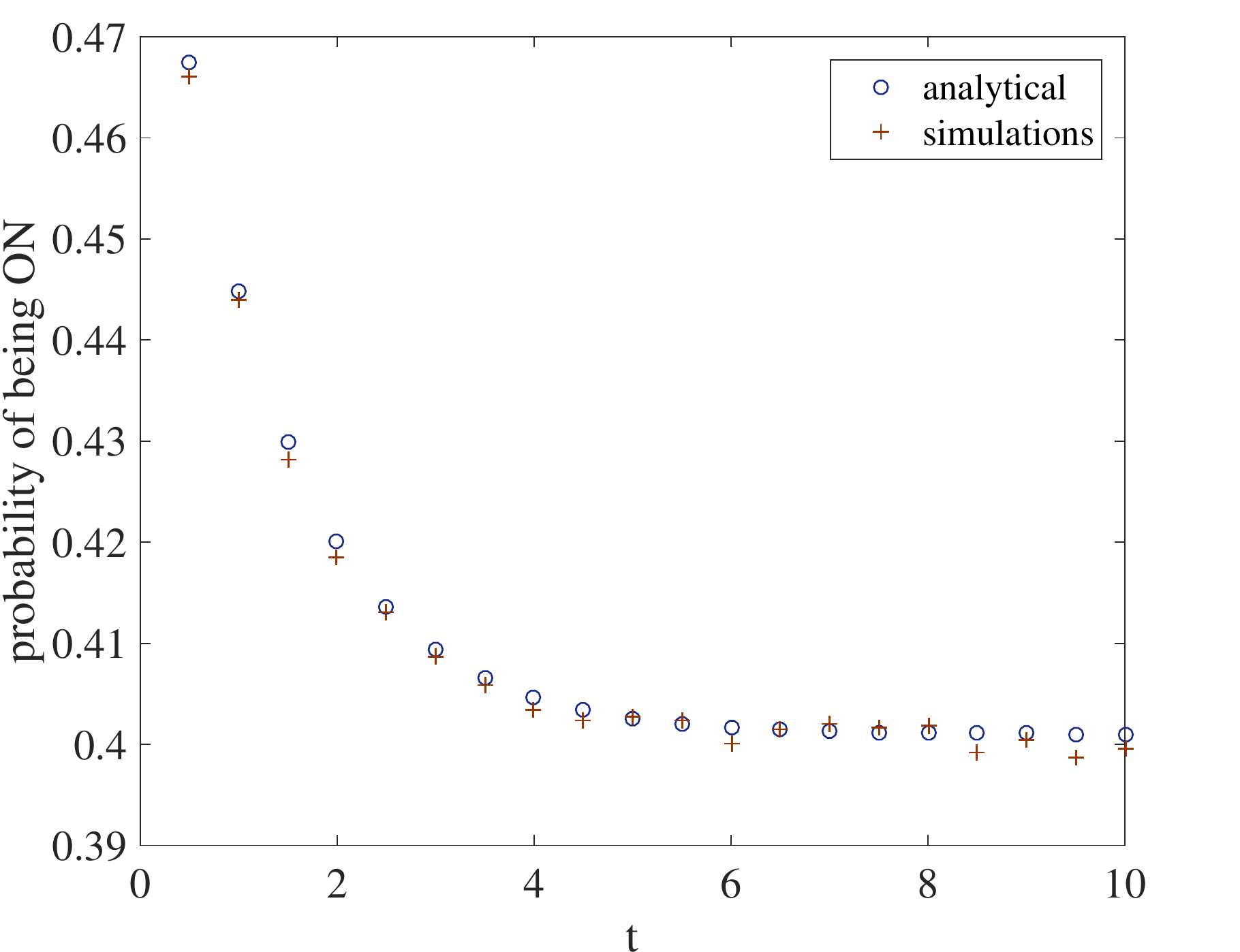}
  \caption{Comparing probability of an edge existing between two specific nodes obtained from simulations and ON-OFF analytical model for exponentially distributed ON-OFF periods} \label{fig:ON-OFF-2}
\end{figure}

Another interesting observation is the relationship between accessibility and delay (or trip duration) between two specific nodes  in the network. We define the delay from node $i$ to $j$ to be the first time slot $t$ that an open temporal path from $i$ to $j$ becomes present starting from $t=1$. We denote this delay by $D_{i \rightarrow j}$.
 The probability of accessibility $P(i\xrightarrow{t} j)$ can be interpreted as the cumulative distribution function (CDF) of the delay. Therefore, the delay probability distribution function (pdf), denoted by $\pi(D_{i\rightarrow j})$, would be immediately available by calculating the derivative of the probability of accessibility, i.e.

\begin{equation}
\pi(D_{i\rightarrow j}) = \frac{dP(i\xrightarrow{t} j)}{dt}
\end{equation}

Therefore, the bounds given in this paper can be similarly differentiated to approximate the density function that describes the delay distribution. Fig. \ref{fig:delay} shows the distribution of delay  for the experiment setting in Fig. \ref{fig:accessibility} as well as approximated distributions obtained by differentiating the lower and  upper bounds.

Another experiment was performed to verify the ON-OFF model derived in Section \ref{sec:ON-OFF}. In our experiment, we assumed that the ON and OFF periods for an edge follow exponential distributions with parameters $\beta^{ON}$ and $\beta^{OFF}$, respectively. Therefore the probability distribution of the total length of being ON (or OFF) given $m$ switchings (changing from ON to OFF or from OFF to ON) is the convolution of $m$ exponentials, i.e., a Gamma distribution $\Gamma(\beta^{ON}, m)$ (or $\Gamma(\beta^{OFF}, m)$). Therefore the pdf of the random variable $S_m^{ON}$ (or $S_m^{OFF}$) is obtained by calculating the convolution of the two Gamma distributions with different parameters. Various methods for such calculations are available in the literature \cite{moschopoulos1985distribution, jasiulewicz2003convolutions}.
\begin{figure}[t]
 \includegraphics[trim = 12mm 70mm 10mm 70mm, clip,scale=.48]{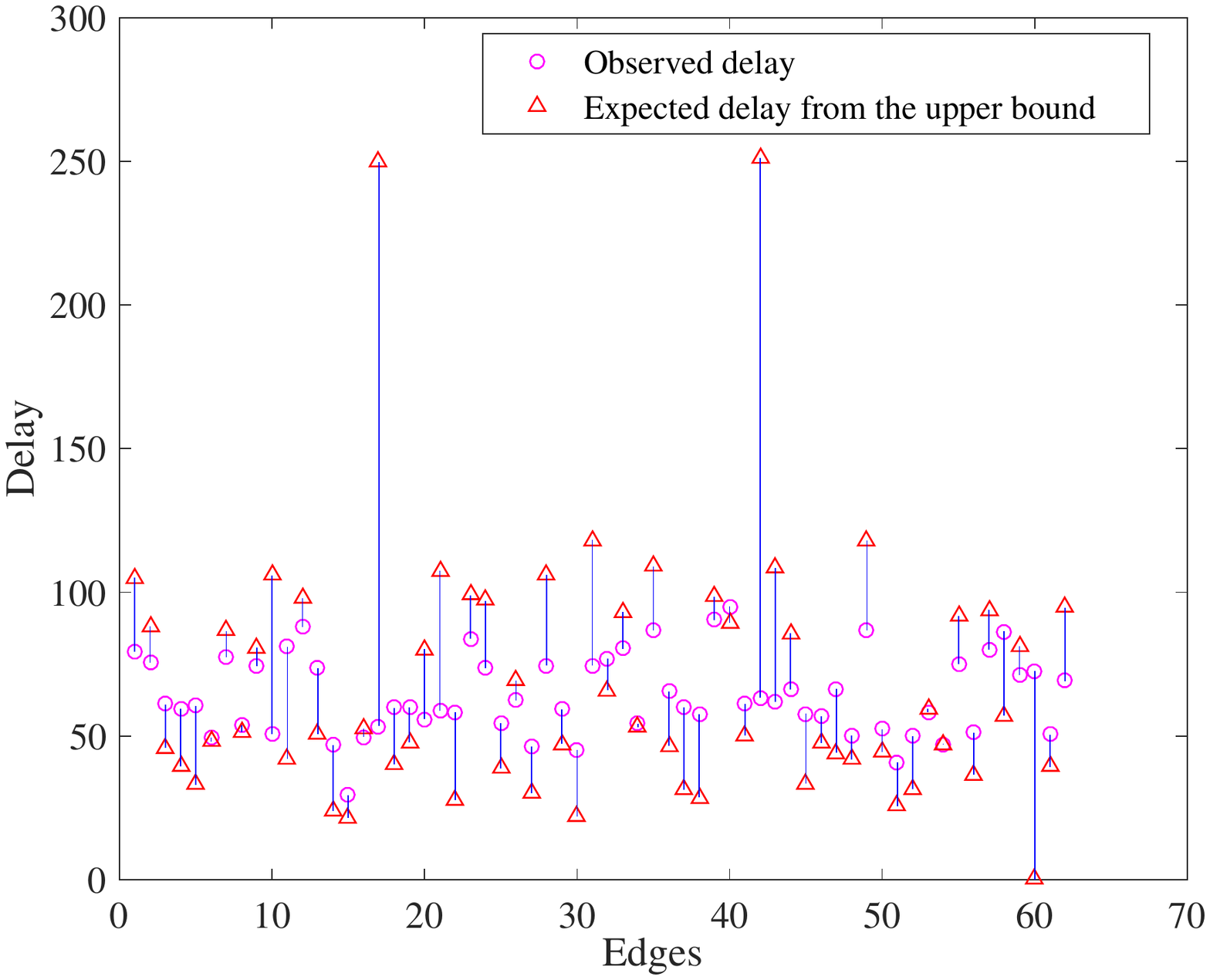}
  \caption{Comparing the expected delay derived from the upper-bound and estimated edge probabilities over the training phase, and the average delay from the $10$ experiments.} \label{fig:delay2}
\end{figure}
\begin{figure}[t]
 \includegraphics[trim = 16mm 83mm 20mm 90mm, clip,scale=.51]{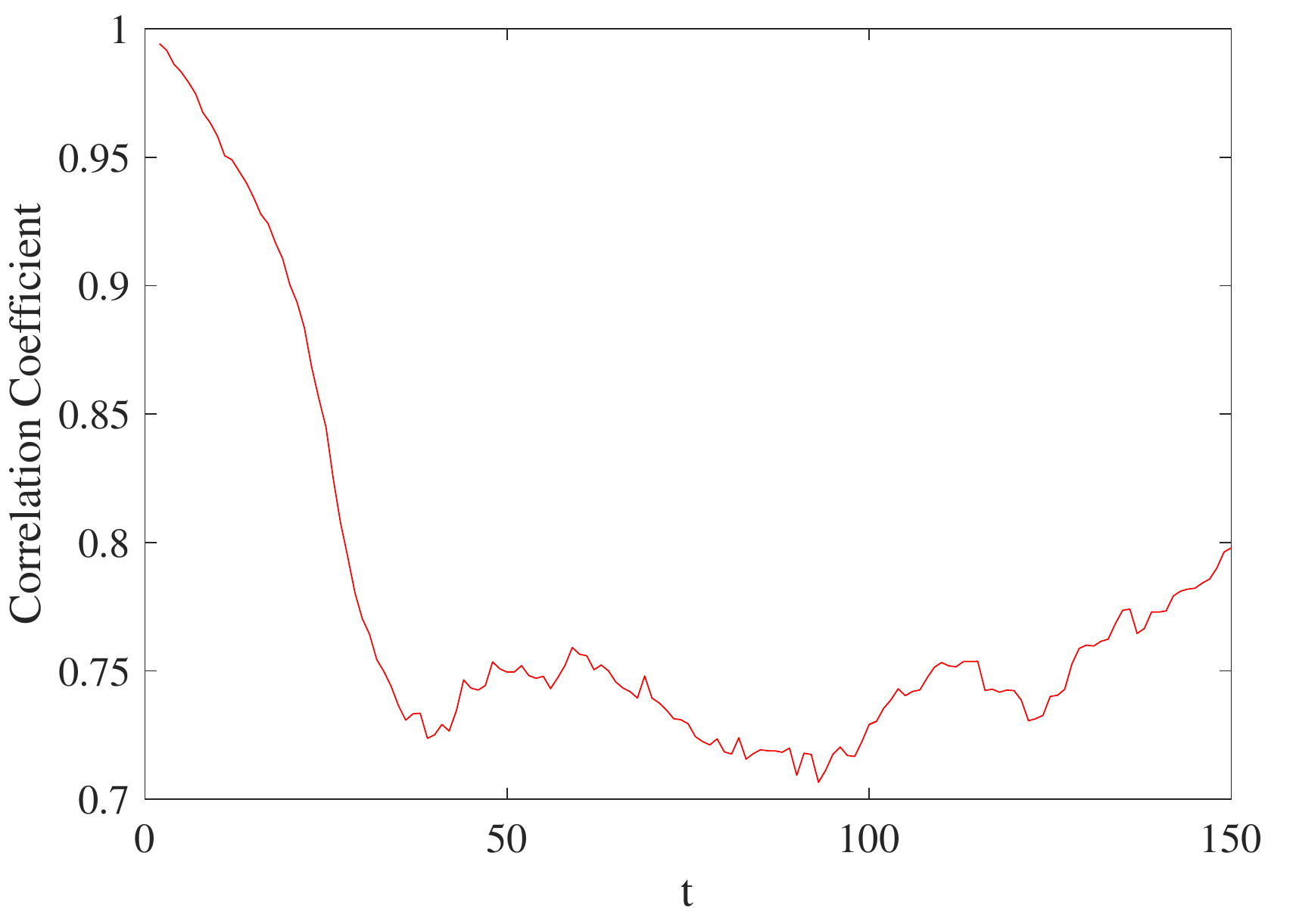}
  \caption{Correlation between the predicted probabilities of accessibility and the empirically estimated values.} \label{fig:correlation}
\end{figure}

 We considered a window observation of $T = 10$ units of time in the continuous domain and observed the status of an edge over the mentioned window. We assumed that the ON and OFF periods are exponentially distributed with $\beta^{ON} = 2$ and $\beta^{OFF} = 3$, respectively. Moreover, we assumed that starting from ON or OFF positions was equally probable (i.e. $p_0 = 0.5$).
 We have run $10000$ trials for this edge and measured the status of the edge (ON or OFF) at the end of every time slot of length $0.5$ from $t = 0$ to $t = 10$ and averaged them over these trials.
 Since in a practical setting the total number of switchings cannot be infinite we limited the number of switchings to be from $m = 0$ to $ m = 40$.
 The comparison between the results of this experiment and the probability obtained from Equation (\ref{eq:ON-OFF}) is shown in Fig. \ref{fig:ON-OFF-2} which verifies the accuracy of the method in Section \ref{sec:ON-OFF} for the mentioned range of $m$.
 \subsection{Real-World Dataset} \label{real}
  As it could be seen from the experiments on synthetically generated networks, the gap between the upper bound and estimated probabilities from Monte-Carlo simulations for networks with randomly assigned probabilities on edges is fairly narrow. Hence, the bound can be naturally nominated as a \emph{predictor} for the probability of accessibility in real-world networks. However, it would be crucial to observe the performance of the given bound beyond the abstractions of the previously mentioned synthetic networks and under more realistic conditions. With this objective, we have used the data collected from the GPS devices of a group of taxis in Rome \cite{roma-taxi-20140717}. A vehicular network has been selected for this experiment as such networks are mathematically more tractable. The reason for this is that in such networks the inter-contact time has been shown to follow an exponential distribution \cite{zhu2010recognizing}. From the entire set of $387$ taxis in the original dataset, we have randomly selected $25$ and observed them over a period of one month (February). The time period of observation has been split to $T = 2688$ time slots of $15$ minutes. We have assumed that two taxis are in contact if their distance is less than $R = 50$ meters at any fraction of time within a given time slot. The occurrence of a contact between two vehicles over a given time slot can be represented by the existence of an edge in the equivalent adjacency graph over that time slot (where each vehicle is represented by a vertex in such a graph).

  We further assume that the duration of a contact is negligible when compared to the length of a time slot. If two vehicles are in contact for a longer period, this can be considered as several consecutive contacts. With this assumption and given the memorylessness of the exponential distribution, the origin of time would not have any impact on the probability of a contact occurring  between two vehicles over a given time slot. If we denote the inter-contact time between two vehicles by random variable $X \sim \lambda e^{-\lambda x}$ and the duration of time slot by $t_0$, the probability of a contact occurring over a given time slot can be approximated by $\int_{0}^{t_0}\lambda e^{-\lambda t}dt = 1-e^{-\lambda t_0}$. It should be noted that here we have assumed that $t_0\ll \frac{1}{\lambda}$. Therefore, it can be assumed that with a high probability at most one contact occurs between the two vehicles. Therefore, transforming the continuous distribution of inter-contact time between the vertices in the corresponding graph would be considerably simpler than the general procedure described in Section \ref{sec:ON-OFF}.

 The experiment is comprised of two phases. In the first phase the distribution of inter-contact time between any two individual vehicles is estimated by fitting to an exponential distribution. In other words, the first phase is used for \emph{training}. We have allocated the period $1\leqslant t \leqslant 1100$ for training. In the second phase we use the distributions obtained from the training to predict the probability of accessibility between pairs of vehicles over the period $1101\leqslant t\leqslant 2600$. We divide this period of $1500$ time slots to $10$ equal and distinct subperiods of $150$ time slots. For each time slot, the number of experiments in which any vehicle $j$ has been accessible from $i$ is assumed to be an estimate of the probability of accessibility from $i$ to $j$.
Fig. \ref{fig:delay2} compares the average delay obtained empirically from the second phase and the predicted expected delay (obtained by using the upper bound as the predictor for accessibility probability) from the training phase for a subset of vertex pairs. We have only selected those pairs that in at least $8$ out of the $10$ experiments accessibility has been established. To avoid a dense figure, half of the pairs have been randomly selected and their delays are compared in the figure.

   Moreover, to evaluate the goodness of the upper bound as a predictor for the accessibility probability, we have measured the correlation coefficient between the vector of all estimated probabilities obtained from the experimental phase and the vector of probabilities predicted by the upper-bound for the entire set of pairs of vertices. Fig. \ref{fig:correlation} represents the variations of the correlation coefficient over time. As it can be observed, even for such a small number of experiments (ten) and for a relatively short phase of training ($1100$ time slots), the correlation coefficient remains above $0.7$ almost all the time. Therefore, the combination of the upper bound (as the predictor) and the contact probability estimation (based on a learning phase) performs with a very good accuracy.

 \section{Conclusion}\label{sec:conclusion}
 Formation of paths in complex networks with time-varying edges where the presence and absence of edges is a function of time and possibly random, is far more complicated than static graphs. In this paper, we studied the formation of such paths and the notion of accessibility between nodes in random temporal networks at a microscopic level. Finding the exact probability of having access from one node to another in such networks is rather complicated. We provided a set of bounds on this probability for a very general setting of probabilities. Moreover, we extended our result to the continuous-time domain networks. We evaluated our analytical results with numerical experiments. The microscopic level analysis given in this paper can be a ground for macroscopic analysis of random temporal networks in future.
\vspace{6mm}
\section*{Acknowledgment}
This work was supported in part by EPSRC grant number EP/N002350/1 (Spatially Embedded Networks).

\nocite{*}

\bibliography{myrefs}

\end{document}